\documentclass[runningheads]{llncs}
\usepackage{comment}
\usepackage{graphicx}
\usepackage{amsfonts}
\usepackage{amsmath}
\usepackage{algorithm}
\usepackage{algpseudocode}
\newcommand{\Endproof}{\hfill$\Box$\\}

\begin{document}
\title{The Fast Algorithm for Online k-server Problem on Trees\thanks{This paper has been supported by the Kazan Federal University Strategic Academic Leadership Program ("PRIORITY-2030").}}
%
%
\author{Kamil Khadiev\inst{1}\orcidID{0000-0002-5151-9908} \and
Maxim Yagafarov\inst{1} }
\authorrunning{K. Khadiev and M. Yagafarov}
%
\institute{Kazan Federal University, Kazan, Russia\\
\email{kamilhadi@gmail.com}}
\maketitle              
\begin{abstract}
  We consider online algorithms for the $k$-server problem on trees. 
There is a $k$-competitive algorithm for this problem, and it is the best competitive ratio. M. Chrobak and L. Larmore provided it. At the same time, the existing implementation has $O(n)$ time complexity for processing a query and $O(n)$ for prepossessing, where $n$ is the number of nodes in a tree. Another implementation of the algorithm has $O(k^2+k\log n)$ time complexity for processing a query and $O(n\log n)$ for prepossessing. 
We provide a new time-efficient implementation of the algorithm. It has $O(n)$ time complexity for preprocessing and $O\left(k(\log n)^2\right)$ for processing a query. 

\end{abstract}
%
%
%
\section{Introduction}
\label{sec:intro}
Online optimization is a field of optimization theory that deals with optimization problems not knowing the future \cite{k2016}. 
The most standard method to define the effectiveness of an online algorithm is a competitive ratio \cite{st85,kmrs86}. The competitive ratio is the worst-case ratio between the cost of a solution found by the algorithm and the cost of an optimal solution.
In the general setting, online algorithms have unlimited computational power. Nevertheless, many papers consider them with different restrictions. Some of them are restrictions on memory \cite{bk2009,gk2015,blm2015,kkm2018,aakv2018,gs93,kk2022,kkzmkry2022,kk2019disj,k2021,kl2020,kk2020}, other ones are restrictions on time complexity \cite{fnn2006,rbm2013}.
This paper focuses on efficient online algorithms in terms of time complexity.
One of the well-known online minimization problems is the $k$-server problem on trees \cite{cl91}. Other related well-known problems are the matching problem, r-gathering problem, facility assignment problem \cite{kp93,an2015,ark2020}.
There is a $k$-competitive deterministic algorithm for the  $k$-server problem on trees, and the algorithm has the best competitive ratio. Expected competitive ratio for a best-known randomized algorithm \cite{bbmn2011,bbmn2015} is $O(\log^3 n\log^2 k)$, where $n$ is the number of nodes in a tree. In this paper, we are focused on the deterministic one. So, the competitive ratio of the deterministic algorithm is the best one. At the same time, the naive implementation has $O(n)$ time complexity for each query and preprocessing. There is a time-efficient algorithm for general graphs \cite{rbm2013} that uses a min-cost-max-flow algorithm, but it is too slow in the case of a tree.  
In the case of a tree, there exists an algorithm with time complexity $O(n\log n)$ for preprocessing and $O\left(k^2+k\log n\right)$ for each query \cite{kkmsy2020}.

We suggest a new time-efficient implementation of the algorithm from \cite{cl91}. It has $O(n)$ time complexity for preprocessing and $O\left(k(\log n)^2\right)$ for processing a query. It is based on data structures and techniques like a segment tree \cite{l2017guide},  heavy-light decomposition (heavy path decomposition) \cite{st83,ht84} for a tree and fast algorithms for computing Lowest Common Ancestor (LCA) \cite{bv93,bfc2000}. 
Let us compare our algorithm with the implementation from \cite{kkmsy2020}. Our prepossessing procedure is more efficient, and we obtain speed-up for the query processing procedure in the case of $k=\omega\left((\log n)^2\right)$.
The algorithm is more efficient than the naive algorithm in the case of $k=o\left(n/(\log n)^2\right)$.

The structure of the paper is as follows. Preliminaries are presented in Section \ref{sec:prelims}. Section \ref{sec:st} contains a subproblem on a segment tree that is used in the main algorithm. The main algorithm is discussed in Section \ref{sec:algorithm}. Section \ref{sec:conclusion} concludes the paper.

\section{Preliminaries}
\label{sec:prelims}
{\bf An online minimization problem} consists of a set $\cal{I}$ of inputs and a cost function. An input is $I = (x_1, \dots , x_n)$, where $n$ is a length of an input $|I|=n$. Furthermore, a set of feasible outputs (or solutions) ${\cal O}(I)$ is associated with each $I$; an output is $O = (y_1, \dots, y_n)$. The cost function assigns a positive real value $cost(I, O)$ to $I\in{ \cal I}$ and $O\in{\cal O}(I)$. The optimal  solution for $I\in{\cal I}$ is $O_{opt}(I)=argmin_{O\in{\cal O}(I)}cost(I,O)$.

Let us define an online algorithm for this problem.
{\bf A deterministic online algorithm}  $A$ computes an output sequence $A(I) = (y_1,\dots , y_n)$ such that $y_i$ is computed based on $x_1, \dots , x_i$.  
We say that $A$ is $c$-{\em competitive} if there exists a constant $\alpha\geq 0$ such that, for every $n$ and for any input $I$ of size $n$, we have: $cost(I,A(I)) \leq c \cdot cost(I,O_{Opt}(I)) + \alpha$. Here, $c$ is the minimal number that satisfies the inequality. Also we call $c$ the {\bf competitive ratio} of $A$. 

\subsection{Graph Theory}\label{sec:graph}
Let us consider a rooted tree $G=(V,E)$, where $V$ is a set of nodes (vertices), and $E$ is a set of edges. Let $n=|V|$ be the number of nodes, and $V=\{v^1,\dots,v^n\}$. Let the root of the tree be the $v^1$ node.

A path $P$ is a sequence of nodes $(v_1,\dots,v_h)$ that are connected by edges, i.e. $(v_i,v_{i+1})\in E$ for all $i\in\{1,\dots,h-1\}$. Note, that there are no duplicates among $v_1,\dots,v_h$. Here, $h$ is the length of the path.  We use $v\in P$ notation if there is $j$ such that $v_j=v$. The notation is reasonable, because there is no duplicates in a path. Note that for any two nodes $u$ and $v$ the path between them is unique because $G$ is a tree.

The distance $dist(v,u)$ between two nodes $v$ and $u$ is the length of the path between them. A height of a node $v$ is the distance from the root that is $dist(v^1,v)$.
For each node $v$ except the root node we can define a parent node $\textsc{Parent}(v)$, it is a node such that $dist(v^1,\textsc{Parent}(v))+1=dist(v^1,v)$ and it belongs to the path from $v^1$ to $v$. We assume that for the root node, $\textsc{Parent}(v^1)=NULL$. Additionally, we can define a set of children $\textsc{Children}(v)=\{u: \textsc{Parent}(u)=v\}$. 

{\bf Distance}$\quad$ For each node $v$ we compute the distance from the $v^1$ (root) node to the node $v$. We call it $dist(v^1,v)$. We can do it using Depth-first search algorithm \cite{cormen2001}.
%
There is a well-known simple algorithm for computing of $dist(v^1,v)$, we present it for completeness in Appendix \ref{apx:comp-dist}. Let $\textsc{ComputeDistance}$ be a subroutine that computes distances. After invocation of this procedure, we can obtain $dist(v^1,v)$ in $O(1)$ time complexity.

{\bf Heavy-Light Decomposition}$\quad$
Heavy-light decomposition is a decomposition of the tree to a set of paths ${\cal P}$. The technique is presented in \cite{ht84,st83}. It has the following properties:

    $\bullet$ Each node $v$ of the tree belongs to exactly one path from ${\cal P}$, i.e., all paths have no intersections, and they cover all nodes of the tree.
    
    $\bullet$ For any node $v$, a path from $v$ to the root of the tree contains nodes of at most $\log_2 n$ paths from ${\cal P}$.
    
    $\bullet$ Let us consider a node $v$ and a path $P\in{\cal P}$ such that $v\in P$. Then, $beg(v)$ is the node that belongs to $P$ and has the minimal height, i.e. $beg(v)$ is such that $dist(v^1,beg(v))=\min\limits_{u\in P}dist(v^1,u)$.
    
    $\bullet$ For a node $v$ of the tree, let $P(v)$ be the path from ${\cal P}$ that contains $v$.
    
    $\bullet$ For a node $v$ of the tree, let $index_P(v)$ be an index of an element of the path $P$. For an index $i$ of an element in the path $P$, let $node_P(i)$ be the node $v$. In other words, if $P=(v_1,\dots,v_h)$, then $v_{index_P(v)}=v$, and $v_i=node_P(i)$ 
     
    $\bullet$  We can construct the set ${\cal P}$ with $O(n)$ time complexity.

{\bf Lowest Common Ancestor}$\quad$
Given two nodes $u$ and $v$ of a rooted tree, the Lowest Common Ancestor is a node $w$ such that $w$ is an ancestor of both $u$ and $v$, and $w$ is the closest one to $u$ and $v$ among all such ancestors. The following result is well-known.

\begin{lemma}[\cite{bv93,bfc2000}]\label{lm:lca}
There is an algorithm for the LCA problem with the following properties:
%
 (i) The time complexity of the preprocessing step is $O(n)$
 %
 (ii)The time complexity of computing LCA for two nodes is $O(1)$.
\end{lemma}

Let $\textsc{LCA\_Preprocessing}()$ be the subroutine for the preprocessing step. Let $\textsc{LCA}(u,v)$ be the procedure for computing LCA of two nodes $u$ and $v$.
We can compute the distance $dist(v,q)$ between nodes $v$ and $q$ using LCA in $O(1)$. Let $l=\textsc{LCA}(v,q)$ be a lowest common ancestor of $v$ and $q$. Then, $dist(v,q)=dist(v^1,q)+dist(v^1,v)-2\cdot dist(v^1,l)$. 
\subsection{$k$-server Problem on a Tree}
We have a rooted tree $G=(V,E)$. We are also given $k$ servers that can move among nodes of $G$. At each time slot, a query $q\in V$
appears, and we have to ``serve'' it, that is, choose one of our servers and
move it to $q$. Other servers are also allowed to move. Our measure of cost is the distance
by which we move our servers.
In other words, if before the query positions of servers are $v_1,\dots,v_k$ and after the query they are $v'_1,\dots,v'_k$, then $q\in\{v'_1,\dots,v'_k\}$ and the cost of the move is $\sum_{i=1}^k dist(v_i,v'_i)$.
The problem is to design a strategy that minimizes the cost
of serving a sequence of queries given online. 
\subsection{Coloration Problem}\label{sec:coloring-def}
Let us present the coloration problem used as a sub-task in the main algorithm for the $k$-server problem. It is used in the following way. In a tree, we color a node $v$ by a color $j$ if the server $ j$ visits the node. More detailed motivation is presented in the next section.  

{\bf Coloration problem.} Assume that we have a sequence of $d$ nodes $v_1,\dots,v_d$ of the tree $G$. We associate a color $c_i$ with a node $v_i$  of the tree $G$, where $0\leq c_i \leq Z$ for some positive integer $Z$. Initially, all nodes are not colored, i.e. $c_i=0$. We should be able to do several operations. Each operation can be one of three types:

    $\bullet$ {\bf Update.} For three integers $l, r, c$ ($1\leq l\leq r\leq d$, $1\leq c\leq Z$), we should color all elements of segment $[l,r]$ by $c$, i.e. $c_i\gets c$ for $l\leq i \leq r$. 
    
    $\bullet$ {\bf Request.} For an integer $x$ ($1\leq x\leq d$), we should return $c_x$. 
    
    $\bullet$ {\bf Request Closest Colored.}  For two integers $l, r$ ($1\leq l\leq r\leq d$), we should return the minimal and the maximal indexes of colored elements from the segment, i.e. the maximal and the minimal $i$ such that $c_i>0$ and $l\leq i\leq r$. 

We can implement these operations using the segment tree data structure \cite{l2017guide}. The definition of the Segment Tree data structure is presented in Section \ref{sec:st}. Assume that we have several procedures. The procedure $\textsc{ConstructST}(1,d)$ constructs a segment tree . This procedure is used by the initialization process for the coloration problem solution (Lemma \ref{lm:st-construct}). The procedure returns the root of the segment tree. The time complexity is $O(d)$. The procedure $\textsc{ColorRequest}(x,root)$ implements the {\bf Request}  operation for $c_x$ and a segment tree with a root node $root$ (Lemma \ref{lm:st-color}). It  has $O(\log d)$ time complexity. The procedure $\textsc{ColorUpdate}(l,r,c, root)$ implements the {\bf Update} operation for a segment $[l,r]$, a color $c$  and a segment tree with a root node $root$ (Lemma \ref{lm:st-update}). It has $O(\log d)$ time complexity.
The procedure $\textsc{GetClosestColorRight}(l,r,root)$ implements the {\bf Request Closest Colored} operation for a segment $[l,r]$ and a segment tree with a root node $root$ (Lemma \ref{lm:st-get-closest}). It returns the minimal index from the segment and has $O(\log d)$ time complexity.
The  $\textsc{GetClosestColorLeft}(l,r,root)$ procedure has similar properties and returns the maximal index. The more detailed discussion is presented in Section \ref{sec:st}.

\section{The Fast Online Algorithm for $k$-server Problem on Trees}\label{sec:algorithm}
Let us describe a $k$-competitive algorithm for the $k$-server problem on trees from \cite{cl91}.

{\bf Chrobak-Larmore's  $k$-competitive algorithm for the $k$-server problem from \cite{cl91}.} Let us have a query $q$, and let servers be in nodes $v_1,\dots,v_k$.
Let a server $i$ be {\em active} if there is no other servers on the path from $v_i$ to $q$. If several severs in a node, then the server with the smallest index is active. Formally, let us consider a path $P=(w_1,\dots,w_h)$, where $w_1=v_i$ and $w_h=q$. Then, $v_{i'}\not\in P$ for $i'\neq i$. If there is $v_{i'}=v_i$, then the server $i$ is active if $i'>i$.
 In each phase, we move each {\em active} server one step towards the node $q$. After each phase, the set of {\em active} servers can be changed.
 We repeat phases (moves of servers) until one of the servers reaches the query node $q$.
The naive implementation of the algorithm has time complexity $O(n)$ for each query. It can be the following. Firstly, we run the Depth-first search algorithm with time labels \cite{cormen2001}. Using it, we can put labels to each node that allows us to check for any two nodes $u$ and $v$, whether $u$ is an ancestor of $v$ in $O(1)$. After that, we can move each active server to the query step by step. Together all active servers cannot visit more than $O(n)$ nodes.

Here, we present an effective implementation of Chrobak-Larmore's algorithm. The algorithm contains two parts that are preprocessing and query processing. The preprocessing part is done once and has $O(n)$ time complexity (Theorem \ref{th:preproc}). The query processing part is done for each query and has $O\left(k(\log n)^2\right)$ time complexity  (Theorem \ref{th:query-proc}).
 
\subsection{Preprocessing}
We do the following steps for the preprocessing:

    $\bullet$ We construct a Heavy-light decomposition ${\cal P}$ for the tree. Properties of decomposition are described in Section \ref{sec:graph}. Assume that,  for construction ${\cal P}$, we have $\textsc{ConstructingHLD}()$ subroutine.
  
        $\bullet$ We do the required preprocessing for the LCA algorithm that is discussed in Section \ref{sec:graph}. Assume that we have $\textsc{LCA\_Preprocessing}()$ subroutine for this procedure.
        
    $\bullet$ For each path $P\in{\cal P}$ we construct a segment tree that will be used for the coloration problem that is described in Section \ref{sec:coloring-def} and Section \ref{sec:st}. Let  $\textsc{ConstructingSegemntTree}(P)$ be a subroutine for construction a segment tree for the path $P$. Let $ST_P$ be a segment tree for the path $P$.  
    
    $\bullet$ Additionally, for each node $v$ we compute the distance from the $v^1$ (root) node to the node $v$ using $\textsc{ComputeDistance}$ subroutine from Section \ref{sec:graph}.

Finally, we have Algorithm \ref{alg:preproc} for the preprocessing.

\begin{algorithm}[ht]
    \caption{$\textsc{Preprocessing}$. Preprocessing procedure.} \label{alg:preproc}
    \begin{algorithmic}
        \State ${\cal P} \gets \textsc{ConstructHLD()}$
        \State $\textsc{LCA\_Preprocessing}()$
        \For{$P \in {\cal P}$}
        \State $ST_P \gets\textsc{ConstructST}(P)$
        \EndFor
        \State $dist(v^1,v^1)\gets 0$, $\textsc{ComputeDistance()}$
    \end{algorithmic}
\end{algorithm}

\vspace{-0.1cm}
\begin{theorem}\label{th:preproc}
Algorithm \ref{alg:preproc} for the preprocessing has time complexity $O(n)$. (See Appendix \ref{apx:th1})
\end{theorem}

\vspace{-0.3cm}
\subsection{Query Processing}
Let us have a query on a node $q$, and servers are in nodes $v_1,\dots,v_k$. We make the following steps:

 {\bf Step 1.} Let us sort all servers by the distance to the node $q$. We assume that if two servers have the same distance to the node $q$, then the server with a smaller index should precede the server with a bigger index.  
 Let $\textsc{Sort}(q,v_1,\dots,v_k)$ be a sorting procedure.
     On the following steps we assume that $dist(v_i,q)\leq dist(v_{i+1},q)$ for $i\in\{1,\dots,k-1\}$.
     
     {\bf Step 2.}  The first server from $v_1$ processes the query. We move them to the node $q$  and color all nodes of a path from $v_1$ to $q$ to color $1$. The node's color shows the number of a server that visited the node. 
     Let the coloring process be implemented as a procedure $\textsc{ColorPath}(v_1,q,1)$. The procedure and detailed description of this step is presented in the end of this section.
     
     {\bf Step 3.} For $i\in\{2,\dots k\}$ we consider a server that stays in the $v_i$ node. It becomes inactive when some other server $j$ becomes closer to the query than $i$-th server. It is easy to see that $j<i$ because the distance from $v_j$ to the target $q$ is smaller than $v_i$ to the same target $q$. When $i$-th server becomes inactive, the server $j$ is active. Therefore, the node where $j$-th server was in that moment should have color $j$.
     
     To obtain the index $j$, we search a colored node closest to $v_i$ on the path from $v_i$ to $q$. The color of this node is $j$. Let the search of the closest colored node be implemented as a procedure $\textsc{GetClosestColor}(v_i,q)$. 
     It is described in the end of this section. 
     Let the obtained node be $w$ and its color is $j$. The $j$-th server reaches the node $w$ in $z = dist(v_j,w)$ steps. After that the $i$-th server becomes inactive. So, we should move the $i$-th server to a node $v'_i$ to $z$ steps on the path from $v_i$ to $w$.  Let the moving process be implemented as a procedure $\textsc{Move}(v_i,w,z)$. It is described in the end of this section 
     . Then, we color all nodes on the path from $v_i$ to $v_i'$ to the color $i$. 
     Let us describe the procedure as Algorithm \ref{alg:query}.

\vspace{-0.2cm}     
    \begin{algorithm}[H]
    \caption{$\textsc{Query}(q)$. Query procedure.} \label{alg:query}
    \begin{algorithmic}
        \State $\textsc{Sort}(q,v_1,\dots,v_k)$ 
        \State $\textsc{ColorPath}(v_1,q,1)$
        \State $v'_1\gets q$
        \For{$i\in\{2,\dots,k\}$}
        \State $(w,j)\gets \textsc{GetClosestColor}(v_i,q)$
        \State $z\gets dist(v_j,w)$
        \State $v'_i \gets \textsc{Move}(v_i,w,z)$
        \State $\textsc{ColorPath}(v_i,v'_i,i)$
        \EndFor
    \end{algorithmic}
\end{algorithm}
\subsubsection{Coloring of a Path}\label{sec:coloring}
Let us consider the problem of coloring nodes on a path from a node $v$ to a node $u$. The color is $c$.

Let $l=LCA(v,u)$ be an LCA of $v$ and $u$. Let   $P_1,\dots,P_t\in {\cal P}$ be paths that contain nodes of the path from $v$ to $l$ and let   $P'_1,\dots,P'_{t'}\in {\cal P}$ be paths that contain nodes of the path from $l$ to $u$.
Let 
$ w_0=v,\textrm{ } w_0\in P_1;
    \quad w_1=beg(P_1),\textrm{ } \textsc{Parent}(w_1)\in P_2;
    \quad w_2=beg(P_2),\textrm{ } \textsc{Parent}(w_2)\in P_3;
    \dots
    w_{t-1}=beg(P_{t-1}),\textrm{ } \textsc{Parent}(w_{t-1})\in P_t;
    \quad w_t=l;
$
and 
%
$ w'_0=u,\textrm{ } w'_0\in P'_1;
    \quad w'_1=beg(P'_1),\textrm{ } \textsc{Parent}(w'_1)\in P'_2;
    \quad w'_2=beg(P'_2),$ $ \textsc{Parent}(w'_2)\in P'_3;
    \dots w'_{t-1}=beg(P'_{t'-1}),\textrm{ } \textsc{Parent}(w'_{t'-1})\in P_t';
    \quad w'_{t'}=l.$
Then, the coloring process is two steps:

$\bullet$ $\textsc{ColorUpdate}(index_{P_i}(\textsc{Parent}(w_{i-1})),index_{P_i}(w_i),c,ST_{P_i})$ for $i\in\{2,\dots,t\}$, and $\textsc{ColorUpdate}(index_{P_i}(w_0),index_{P_i}(w_1),c,ST_{P_1})$;

$\bullet$ $\textsc{ColorUpdate}(index_{P'_i}(\textsc{Parent}(w'_{i-1})),index_{P'_i}(w'_o),c,ST_{P'_i})$ for $i\in\{2,\dots,t'\}$, and $\textsc{ColorUpdate}(index_{P'_i}(w'_0),index_{P'_i}(w'_1),c,ST_{P'_1})$.
The procedure is presented as Algorithm \ref{alg:coloring}.

\vspace{-0.2cm}
    \begin{algorithm}[H]
    \caption{$\textsc{ColorPath}(v,u,c)$. Coloring the path between $v$ and $u$.} \label{alg:coloring}
    \begin{algorithmic}
        \State $l \gets \textsc{LCA}(v,u), \quad w\gets v, \quad P\gets P(v)$
        \While{$P\neq P(l)$}
        \State $bw \gets beg(P)$
        \State $\textsc{ColorUpdate}(index_P(w),index_P(bw),c,ST_{P})$
        \State $w\gets \textsc{Parent}(bw), \quad P\gets P(w)$
        \EndWhile
        \State $\textsc{ColorUpdate}(index_P(w),index_P(l),c,ST_{P})$
        
        \State $w\gets u, \quad P\gets P(u)$
        \While{$P\neq P(l)$}
        \State $bw \gets beg(P)$
        \State $\textsc{ColorUpdate}(index_P(w),index_P(bw),c,ST_{P})$
        \State $w\gets \textsc{Parent}(bw), \quad P\gets P(w)$
        \EndWhile
        \State $\textsc{ColorUpdate}(index_P(w),index_P(l),c,ST_{P})$
        
    \end{algorithmic}
\end{algorithm}

\begin{lemma}\label{lm:coloring}
Time complexity of Algorithm \ref{alg:coloring} is $O\left((\log n)^2\right)$.(See Appendix \ref{apx:coloring})
\end{lemma}

\subsubsection{The Search of the Closest Colored Node}\label{sec:get-color}
Let us consider the problem of searching the closest colored node on the path from $v$ to $u$. The idea is similar to the idea from the previous section.
Let $l=LCA(v,u)$ be a LCA of $v$ and $u$. Let   $P_1,\dots,P_t\in {\cal P}$ be paths that contain nodes of a path from $v$ to $l$ and let   $P'_1,\dots,P'_{t'}\in {\cal P}$ be paths that contain nodes of a path from $l$ to $u$.
Let
$ w_0=v,\textrm{ } w_0\in P_1;
    \textrm{ } w_1=beg(P_1),\textrm{ } \textsc{Parent}(w_1)\in P_2;
    \textrm{ } w_2=beg(P_2),\textrm{ } \textsc{Parent}(w_2)\in P_3;
    \dots;
        w_{t-1}=beg(P_{t-1}),\textrm{ } \textsc{Parent}(w_{t-1})\in P_t;
    \textrm{ } w_t=l;
 $ and 
$ w'_0=u,\textrm{ } w'_0\in P'_1;
    \textrm{ } w'_1=beg(P'_1),\textrm{ } \textsc{Parent}(w'_1)\in P'_2;
    \textrm{ } w'_2=beg(P'_2),\textrm{ } \textsc{Parent}(w'_2)\in P'_3;
    \dots;
    w'_{t-1}=beg(P'_{t'-1}),\textrm{ } \textsc{Parent}(w'_{t'-1})\in P_t';
    \textrm{ } w'_{t'}=l.
 $ 
For the searching process, firstly,
we invoke the procedure $\textsc{GetClosestColorRight}(index_{P_i}(w_0),index_{P_i}(w_i),ST_{P_i})$. Assume that the procedure returns $NULL$ if there are no colored nodes in the segment. If the procedure returns $NULL$, then we invoke the procedure $\textsc{GetClosestColorRight}(index_{P_i}(\textsc{Parent}(w_{i-1})),index_{P_i}(w_i),ST_{P_i})$ for $i\in\{2,\dots,t\}$. We stop on the minimal $i$ such that the result is not $NULL$. If all of them are $NULL$, then we continue. 
Then, we invoke the procedure

\noindent
$\textsc{GetClosestColorLeft}(index_{P'_i}(\textsc{Parent}(w'_{i-1})),index_{P'_i}(w_i),ST_{P'_i})$
for $i\in\{t',\dots,2\}$. We stop on the maximal $i$ such that a result is not $NULL$. If all of them $NULL$, then we invoke $\textsc{GetClosestColorLeft}(index_{P'_1}(w'_{0})),index_{P'_1}(w_1),ST_{P'_1})$.
The procedure is presented as Algorithm \ref{alg:close-color}.
    \begin{algorithm}[H]
    \caption{$\textsc{GetClosestColor}(v,u)$. Getting the closest colored vertex on the path between $v$ and $u$.} \label{alg:close-color}
    \begin{algorithmic}
        \State $l \gets \textsc{LCA}(v,u),\quad w\gets v,\quad P\gets P(v),\quad g\gets NULL$
        \While{$g=NULL$ and $P\neq P(l)$}
         \State $bw \gets beg(P)$
        \State $g\gets\textsc{GetClosestColorRight}(index_P(w),index_P(bw),ST_{P})$
        \If{$g=NULL$}
        \State $w\gets \textsc{Parent}(bw), \quad P\gets P(w)$
        \EndIf
        \EndWhile
        \If{$g=NULL$}
         \State $g\gets\textsc{GetClosestColorRight}(index_P(w),index_P(l),ST_{P})$
         \EndIf
         \If{$g=NULL$}
          \State $i\gets 0, \quad w'_i\gets u, \quad P\gets P(u)$
       
        \While{$P\neq P(l)$}
        \State $i\gets i+1, \quad w'_i \gets beg(P), \quad bw\gets \textsc{Parent}(w'_i), \quad P\gets P(bw)$
        \EndWhile
            \State $g\gets\textsc{GetClosestColorLeft}(index_P(\textsc{Parent}(w'_i)),index_P(l),ST_{P})$
        \While{$g=NULL$}
         \State $P\gets P(w_i), \quad bw \gets \textsc{Parent}(w'_{i-1})$
        \State $g\gets\textsc{GetClosestColorLeft}(index_P(bw),index_P(w'_i),ST_{P})$
        \State $i\gets i-1$
        \EndWhile
        \EndIf
        \State $resW\gets node_P(g), \quad j\gets \textsc{ColorRequest}(g,ST_P)$
        \State\Return{$(resW,j)$}
    \end{algorithmic}
\end{algorithm}

Let us discuss time complexity of the algorithm.
\begin{lemma}\label{lm:close-color}
Time complexity of Algorithm \ref{alg:close-color} is $O\left((\log n)^2\right)$.
\end{lemma}
\begin{proof}
Due to properties of Heavy-light decomposition from Section \ref{sec:graph}, we have $t, t'= O(\log n)$. Due to results from Section \ref{sec:coloring-def}, each invocation of $\textsc{GetClosestColorLeft}$ or $\textsc{GetClosestColorRight}$ for $P$ has time complexity $O(\log |P|)=O(\log n)$. So, the total time complexity is $O\left((\log n)^2 \right)$. 
\Endproof\end{proof}
\subsubsection{Moving of a Server}\label{sec:move}
Let us consider a moving of a server from $v$ to a distance $g$ on a path from $v$ to $u$. The idea is similar to the idea from the previous section.
Let $l=LCA(v,u)$ be a LCA of $v$ and $u$. Let  $P_1,\dots,P_t\in {\cal P}$ be paths that contains nodes of a path from $v$ to $l$ and let   $P'_1,\dots,P'_{t'}\in {\cal P}$ be paths that contains nodes of a path from $l$ to $u$.
Let 
$ w_0=v,\textrm{ } w_0\in P_1;
    \textrm{ } w_1=beg(P_1),\textrm{ } \textsc{Parent}(w_1)\in P_2;
    \textrm{ } w_2=beg(P_2),\textrm{ } \textsc{Parent}(w_2)\in P_3;
    \dots;
        w_{t-1}=beg(P_{t-1}),\textrm{ } \textsc{Parent}(w_{t-1})\in P_t;
    \textrm{ } w_t=l;
 $ and 
$ w'_0=u,\textrm{ } w'_0\in P'_1;
    \textrm{ } w'_1=beg(P'_1),\textrm{ } \textsc{Parent}(w'_1)\in P'_2;
    \textrm{ } w'_2=beg(P'_2),\textrm{ } \textsc{Parent}(w'_2)\in P'_3;
    \dots;
    w'_{t-1}=beg(P'_{t'-1}),\textrm{ } \textsc{Parent}(w'_{t'-1})\in P_t';
    \textrm{ } w'_{t'}=l.
 $ 
Then, the moving process is the following.
We check whether distance $dist(\textsc{Parent}(w_{i-1}),w_{i})\leq g$. If $dist(\textsc{Parent}(w_{i-1}),w_{i}))\leq g$, then we can return the node $node_{P_i}(index_{P_i}(\textsc{Parent}(w_{i-1}))+g)$ as a result and stop the process. Otherwise, we reduce $g\gets g-dist(\textsc{Parent}(w_{i-1}),w_{i})-1$ and move to the next $i$, i.e. $i\gets i+1$. We do it for $i\in\{1,\dots,t\}$.

If $g>0$, then we continue with the path from $l$ to $u$. 
We check whether  $dist(\textsc{Parent}(w'_{i-1}),w'_{i})\leq g$. If $dist(\textsc{Parent}(w'_{i-1}),w'_{i}))\leq g$, then we can return the node $node_{P'_i}(index_{P'_i}(w'_i)-g)$ as a result and stop the process. Otherwise, we reduce $t\gets t-dist(\textsc{Parent}(w'_{i-1}),w'_{i})-1$ and move to the previous $i$, i.e. $i\gets i-1$. We do it for $i\in\{t',\dots,1\}$.

\begin{lemma}\label{lm:move}
Time complexity of the moving is $O\left(\log n\right)$.
\end{lemma}
\begin{proof}
Due to properties of Heavy-light decomposition from Section \ref{sec:graph}, we have $t, t'= O(\log n)$. The time complexity for processing of each path is $O(1)$. So, the total time complexity is $O\left(\log n\right)$. 
\Endproof\end{proof}
\subsubsection{Correctness and Complexity of the Query Processing}\label{sec:query-compl}
\begin{theorem}\label{th:query-proc}
The query processing Algorithm \ref{alg:query}  has time complexity $O\left(k(\log n)^2\right)$.
\end{theorem}
\begin{proof}
The complexity of servers sorting by distance is $O(k\log k)$. Due to Lemma \ref{lm:coloring}, Lemma \ref{lm:close-color} and Lemma \ref{lm:move}, the complexity for processing one server is $O\left(\log n + (\log n)^2+(\log n)^2 \right)=O\left( (\log n)^2 \right)$. So, the total complexity of processing all servers is $O\left(k\log k +k(\log n)^2\right)=O\left(k(\log n)^2\right)$ because $k<n$. 
\Endproof\end{proof}

\section{Segment Tree with Range Updates for Coloration Problem}\label{sec:st}
In the paper, we use a segment tree with range updates for Coloration Problem (Section \ref{sec:coloring-def}) as one of the main tools for the main algorithm. The data structure allows us to do the main operations for the coloration problem with logarithmic time complexity. As a book with a description of the data structure \cite{l2017guide} can be used.

Firstly, let us describe the segment tree data structure. It is the full binary tree of height $h$ such that $2^{h-1}<d\leq 2^h$. The data structure works with the sequence of elements of the length $2^h$, but we are care only about the first $d$ elements.
Each node of the tree is associated with some segment $[a,b]$ such that $1\leq a\leq b\leq 2^h$.   
Each leaf is associated with elements of the sequence or we can say that it is associated with a segment of size $1$. $i$-th node of the last level is associated with a segment $[i,i]$. Let us consider an internal node $v$ and its two children $u$ and $w$. Then, $u$ is associated with a segment $[a,q]$, $w$ is associated with a segment $[q+1,b]$, and $v$ is associated with a segment $[a,b]$ for some $1\leq a\leq q<b\leq 2^h$. Note that because of the structure of the tree, we have $q=(a+b)/2$.

Each node $v$ of the segment tree is labeled by a color $C(v)$, where $0\leq C(v)\leq Z$. Assume that $v$ is associated with a segment $[a,b]$. If $C(v)=0$, then the segment $[a,b]$  is not colored at all or it has not a single color. If $1\leq C(v)\leq Z$, then the segment has a single color $C(v)$, i.e. $c_a=C(v),\dots, c_b=C(v)$.
Additionally, we add two labels $Max(v)$ and $Min(v)$. $a\leq Max(v) \leq b$ is the maximal index of a colored element of the segment. $a\leq Min(v) \leq b$ is the minimal index of a colored element of the segment. Initially, $Max(v)\gets-1$, $Min(v)\gets 2^h+1$.

For a node $v$ and the associated segment $[a,b]$, we use the following notation. 
    
    $\bullet$ $\textsc{Left}(v)$ is the left border of the segment. $\textsc{Left}(v)=a$
    
    $\bullet$ $\textsc{Right}(v)$ is the right border of the segment. $\textsc{Right}(v)=b$
    
    $\bullet$ $\textsc{LeftChild}(v)$ is the left child of $v$.
    
    $\bullet$ $\textsc{RightChild}(v)$ is the right child of $v$.
 
 Let $\textsc{ConstructST}(a,b)$ be a procedure that returns the root of a segment tree for a segment $[a,b]$. The procedure is standard and has a property that described in Lemma \ref{lm:st-construct}. Let us present the description and the proof of the lemma in Appendix \ref{apx:st-constract} for completeness.
\begin{lemma}\label{lm:st-construct}
Time complexity of the  segment tree constructing procedure $\textsc{ConstructST}(1,d)$ is $O(d)$.
\end{lemma}

Let us describe the processing of three types of operations.

{\bf Request.} The operation is requesting $c_x$ for some $1\leq x\leq 2^h$. We start with the root node of the segment tree. Assume that we observe a node $v$. If $C(v)=0$, then we go to the child that is associated with a segment $[a,b]$, where $a\leq x \leq b$. We continue this process until we meet $v$ such that $C(v)\geq 1$ or $v$ is a leaf. If $C(v)\geq 1$, then the result is $C(v)$. If $C(v)=0$ and $v$ is a leaf, then $c_x$ is not assigned yet. Let a name of the procedure be $\textsc{ColorRequest}(x,root)$. It is a request for a color $c_x$ from a segment tree with $root$ node as the root. If the color is not assigned, then the procedure returns $0$. The implementation of the procedure is presented in Appendix \ref{apx:st-update} and properties are presented in the following lemma. 
\begin{lemma}\label{lm:st-color}
The request color procedure $\textsc{ColorRequest}$ is correct and has $O(\log d)$ time complexity. (See Appendix \ref{apx:st-proofs})
\end{lemma}

{\bf Update.} Assume that we want to color a segment $[l,r]$ in a color $c$ , where $1\leq c\leq Z$ , $1\leq l\leq r\leq 2^h$. 
Let us describe two specific cases. The first one is the coloring of a prefix and the second one is the coloring of a suffix.
Let us have a segment tree with the root node $root$. Let us consider the general case, where a segment $[q,t]$ is associated with the node $root$.

Firstly, assume that $[l,r]$ is a prefix of $[q,t]$, i.e. $q=l$ and $q\leq r \leq t$.
Let us observe a node $v$ and an associated segment $[a,b]$. If $v$ is a leaf, then we assign $C(v)\gets c$ and stop. Otherwise, we continue. We use a variable $c'$ for an existing color. Initially $c'\gets 0$. If on some step $C(v)\geq 1$ and $c'=0$, then we assign $c'\gets C(v)$. If $c'\geq 1$ or $C(v)=0$, then we do not change $c'$ because we already have a color for the segment from an ancestor.

Let $u$ be the left child of $v$, and let $w$ be the right child of $v$. We update $Max(v)\gets\max(Max(v),r)$, $Min(v)\gets \textsc{Left}(v)$ because $[l,r]$ is a prefix. Then, we do the following action.
    
    $\bullet$ If $r\in [a, (a+b)/2]$, then we go to the left child $u$. Additionally, if $c'\geq 1$, then we color $C(w)\gets c'$ because a segment of $w$ has no intersection with $[l,r]$ and keeps its color $c'$.
    
    $\bullet$ If $r\in [(a+b)/2+1, b]$, then we go to the right child $w$. Additionally, we color $C(u)\gets c$ and update $Min(u)\gets\textsc{Left}(u)$, $Max(u)\gets\textsc{Right}(u)$ because $[a, (a+b)/2]$ of $u$ is a subsegment of $[l,r]$. Additionally, we update $l\gets (a+b)/2+1$ because the segment $[a,(a+b)/2]$ is colored and the segment $[(a+b)/2+1,r]$ is left. The new segment is a prefix of the segment tree with the root node $w$.

 Let us call the procedure $\textsc{ColorUpdatePrefix}$ and present it in Appendix \ref{apx:st-algos}.

Secondly, assume that $[l,r]$ is a suffix of $[q,t]$, i.e. $t=r$ and $q\leq l \leq t$. This function is similar to the previous one. The difference is the following.
Let $u$ be the left child of $v$, and let $w$ be the right child of $v$.  We update $Min(v)\gets\min(Min(v),l)$, $Max(v)\gets \textsc{Right}(v)$ because $[l,r]$ is a suffix. Then, we do the following action.

    $\bullet$ If $l\in [(a+b)/2+1, b]$, then we go to the right child $w$. Additionally, if $c'\geq 1$, then we color $C(u)\gets c'$ because a segment of $u$ has no intersection with $[l,r]$ and we keep its color $c'$.
    
    $\bullet$ If $l\in [a, (a+b)/2]$, then we go to the left child $u$. Additionally, we color $C(w)\gets c$ and update $Min(w)\gets\textsc{Left}(w)$, $Max(w)\gets\textsc{Right}(w)$ because $[(a+b)/2+1,b]$ of $w$ is a subsegment of $[l,r]$. Additionally, we update $r\gets (a+b)/2$ because the segment $[(a+b)/2+1,b]$ is colored and the segment $[l, (a+b)/2]$ is left. The new segment is a suffix of the segment tree with the root node $u$. 

 Let us call this procedure $\textsc{ColorUpdateSuffix}$ and present it in Appendix \ref{apx:st-algos}. 

Finally, let us consider a general case for $[l,r]$, i.e. $q\leq l \leq r \leq t$.
Assume that we observe a node $v$ and an associated segment $[a,b]$. If $v$ is a leaf, then we assign $C(v)\gets c$ and stop. Otherwise, we continue. We use a variable $c'$ for an existing color. Initially $c'\gets 0$. If on some step $C(v)\geq 1$ and $c'=0$, then we assign $c'\gets C(v)$. If $c'\geq 1$ or $C(v)=0$, then we do not change $c'$.
We update $Min(v)\gets\min(Min(v),l)$, $Max(v)\gets\max(Max(v),r)$. 

    $\bullet$ If $(a+b)/2+1 \leq l\leq r\leq b$, then we go to the right child $w$. Additionally, if $c'\geq 1$, then we color $C(u)\gets c'$ because a segment of $u$ has no intersection with $[l,r]$ and keeps its color $c'$.
    
    $\bullet$ If $a \leq l\leq r\leq (a+b)/2$, then we go to the left child $u$. Additionally, if $c'\geq 1$, then we color $C(w)\gets c'$ because a segment of $w$ has no intersection with $[l,r]$ and we keep its color $c'$.
    
    $\bullet$ If $a \leq l\leq (a+b)/2 \leq r\leq b $, then we split our segment to $[l,(a+b)/2]$ and $[(a+b)/2+1,r]$. The segment $[l,(a+b)/2]$ is a suffix of the segment tree with the root $u$. For coloring it, we invoke $\textsc{ColorUpdateSuffix}(l,(a+b)/2,c,c',u)$. The segment $[(a+b)/2+1,r]$ is a prefix of the segment tree with the root $w$. For coloring it, we invoke $\textsc{ColorUpdatePrefix}((a+b)/2+1,r,c,c',w)$.

 Let us call the procedure $\textsc{ColorUpdate}$ and present it in Appendix \ref{apx:st-algos}.
\begin{lemma}\label{lm:st-update}
The update procedure $\textsc{ColorUpdate}$ is correct and has $O(\log d)$ time complexity.
(See Appendix \ref{apx:st-proofs})
\end{lemma}

{\bf Request the Closest Colored Element.}
 Assume that we want to get the minimal index of a colored element from a segment $[l,r]$, where  $1\leq l\leq r\leq 2^h$. Let $[q,t]$ be a segment of the root of the segment tree.
Let us describe two specific cases that are requesting from a prefix of $[q,t]$ and requesting from a suffix of $[q,t]$.

Firstly, assume that $[l,r]$ is a prefix of $[q,t]$, i.e. $q=l$ and $q\leq r \leq t$. Assume that we observe a node $v$ and an associated segment $[a,b]$. Let $u$ be the left child of $v$, and let $w$ be the right child of $v$. We do the following action.
    
    $\bullet$ If $r\leq(a+b)/2$,  then we go to the left child $u$. 
    
    $\bullet$ If $r>(a+b)/2$ and $Min(u)$ is assigned (i.e there are colored elements in the left child $u$), then the result is $Min(u)$ and we stop the process.
    
    $\bullet$ If $r>(a+b)/2$ and $Min(u)$ is not assigned (i.e there is no colored element in the left child $u$), then we go to the right child $w$. 

If there is no colored elements in $v$, then the algorithm returns $NULL$. We call the procedure $\textsc{GetClosestColorRightPrefix}(l,r,root)$ and present it in Appendix \ref{apx:st-algos}.

Secondly, assume that $[l,r]$ is a suffix of $[q,t]$, i.e. $t=r$ and $q\leq l \leq t$. Assume that we observe a node $v$ and an associated segment $[a,b]$. Let $u$ be the left child of $v$, and let $w$ be the right child of $v$. We do the following action.

$\bullet$ If $l\geq(a+b)/2+1$, then we go to the right child $w$. 

$\bullet$ If $l\leq(a+b)/2$ and $Min(u)$ is assigned (i.e there are colored elements in the left child $u$), then we go to the left child $u$. 
    
$\bullet$ If $l\leq(a+b)/2$ and $Min(u)$ is not assigned (i.e there is no colored element in the left child $u$), then the result is $Min(w)$ and we stop the process.
    
 If there is no colored elements in $v$, then the algorithm returns $NULL$. We call the procedure $\textsc{GetClosestColorRightSuffix}(l,r,root)$ and present it in Appendix \ref{apx:st-algos}.

Finally, let us consider the general case, i.e. $q\leq l \leq r\leq t$. Assume that we observe a node $v$ and an associated segment $[a,b]$. Let $u$ be the left child of $v$, and let $w$ be the right child of $v$. We do the following action.

    $\bullet$ If $(a+b)/2+1 \leq l\leq r\leq b$, then we go to the right child $w$. 
  
  $\bullet$ If $a \leq l\leq r\leq (a+b)/2$, then we go to the left child $u$.
  
    $\bullet$ If $a \leq l\leq (a+b)/2 \leq r\leq b $, then we split our segment to $[l,(a+b)/2]$ and $[(a+b)/2+1,r]$. The segment $[l,(a+b)/2]$ is a suffix of the segment tree with the root $u$. We  invoke $\textsc{GetClosestColorRightSuffix}(l,(a+b)/2,u)$. If the result is not $NULL$, then there is a colored element in the left child, and we return the result of the procedure. If the result is $NULL$, then there is no colored element in the left children, and only the right children can have the minimal colored element. So, we invoke $\textsc{GetClosestColorRightPrefix}((a+b)/2+1,r,u)$ and we return the result of the procedure.

If there are no colored elements in $v$, then the algorithm returns $NULL$.

We call this function $\textsc{GetClosestColorRight}(l,r,root)$.
We can define the function that returns the maximal index of a colored element symmetrically. We call it $\textsc{GetClosestColorLeft}(l,r,root)$.
\begin{lemma}\label{lm:st-get-closest}
The request the closest colored element procedures $\textsc{GetClosestColorLeft}$ and $\textsc{GetClosestColorRight}$ are correct and have $O(\log d)$ time complexity. (See Appendix \ref{apx:st-proofs})
\end{lemma}

\section{Conclusion}\label{sec:conclusion}
We discuss the time-efficient implementation of online algorithms for the $k$-server problem on trees. Here we present an algorithm with $O(n)$ time complexity for preprocessing and $O(k(\log n)^2)$ time complexity for processing a query. 
It process a query faster than existing implementations of \cite{cl91} (the naive implementation and \cite{kkmsy2020}) in a case of  $\omega\left((\log n)^2\right)=k=o\left(n/(\log n)^2\right)$. This case is reasonable in practice. 
An open problem is an analysis of time complexity for the randomized algorithm \cite{bbmn2011,bbmn2015}.




%
%
%
%

\bibliographystyle{splncs04}
\bibliography{tcs}
\newpage
\appendix
\section{Segment Tree Constructing Procedure}\label{apx:st-constract}
We can construct a segment tree using a simple recursive procedure. Let $\textsc{ConstructST}(a,b)$ be a procedure that returns the root of a segment tree for a segment $[a,b]$.  We present it in Algorithm \ref{alg:st-construct}.

\begin{algorithm}[ht]
    \caption{$\textsc{ConstructST}(a,b)$. A procedure for constructing a segemtn tree for a segment $[a,b]$} \label{alg:st-construct}
    \begin{algorithmic}
       \State $v\gets$ a new node
       \State $\textsc{Left}(v)\gets a$
       \State $\textsc{Right}(v)\gets b$
       \State $C(v)\gets 0$, $Max(v)\gets-1$, $Min(v)\gets 2^h+1$\Comment{The values for $Min(v)$ and $Max(v)$ mean they are not assigned.}
       \If{$a\neq b$}\Comment{not a leaf}
        \State $\textsc{LeftChild}(v)\gets \textsc{ConstructST}(a,(a+b)/2)$
        \State $\textsc{RightChild}(v)\gets \textsc{ConstructST}((a+b)/2+1, b)$
       \EndIf
    \end{algorithmic}
\end{algorithm}
Let us discuss a property of Algorithm \ref{alg:st-construct}.

\noindent
{\bf Lemma \ref{lm:st-construct}}{\em
 Time complexity of Algorithm \ref{alg:st-construct} is $O(d)$.}
\begin{proof}
We construct each node in $O(1)$. A number of nodes on each next level is twice bigger compering to the previous one. The number of nodes on a level $i$ is $2^i$. The number of levels is $h=\lceil\log_2 d\rceil$. So, the total time complexity is $O(\sum_{i=0}^{h} 2^i)=O(2^{h})=O(d)$.
\Endproof\end{proof}

\section{Request Operation for a Segment Tree.}\label{apx:st-update} Assume that we want to get $c_x$ for some $1\leq x\leq 2^h$. We start with the root node. Assume that we observe a node $v$. If $C(v)=0$, then we go to the child that is associated with a segment $[a,b]$, where $a\leq x \leq b$. We continue this process until we meet $v$ such that $C(v)\geq 1$ or $v$ is a leaf. If $C(v)\geq 1$, then the result is $C(v)$. If $C(v)=0$ and $v$ is a leaf, then $c_x$ is not assigned. Let us describe this procedure in Algorithm \ref{alg:st-request}.
\begin{algorithm}[H]
    \caption{$\textsc{ColorRequest}(x,root)$. A request for a color $c_x$ from a segment tree with $root$ node as a root. If the color is not assignet, then the procedure returns $0$.} \label{alg:st-request}
    \begin{algorithmic}
        \State $v\gets root$
        \While{$v$ is not a leaf and $C(v)=0$} 
        \State $u\gets \textsc{LeftChild}(v)$
        \State $w\gets \textsc{RightChild}(v)$
        \If{$x\leq \textsc{Right}(u)$}
        \State $v\gets u$
        \EndIf
        \If{$x> \textsc{Right}(u)$}
        \State $v\gets w$
        \EndIf
        \EndWhile
        \State \Return{$C(v)$}.
    \end{algorithmic}
\end{algorithm}

\section{An Implementation of $\textsc{ComputeDistance}$ Procedure}\label{apx:comp-dist}

Let us consider a node $u$ and a set of children of the node $\textsc{Children}(u)$. Then, for any $v\in \textsc{Children}(u)$, we have $dist(v^1,v)=dist(v^1,u)+1$. Additionally, $dist(v^1,v^1)=0$. So, the distance computing is presented in recursive Algorithms \ref{alg:dist} and \ref{alg:dist-rec}

\begin{algorithm}[ht]
    \caption{$\textsc{ComputeDistance}()$. Computing distances from a root.} \label{alg:dist}
      \begin{algorithmic}
        \State $dist(v^1,v^1)\gets 0$
        \State $\textsc{ComputeDistanceRec}(1)$ 
    \end{algorithmic}
\end{algorithm}

\begin{algorithm}[ht]
    \caption{$\textsc{ComputeDistanceRec}(u)$. Computing distance to a node $u$.} \label{alg:dist-rec}
    \begin{algorithmic}
        \For{$v \in \textsc{Children}(u)$}
        \State $dist(v^1,v)\gets dist(v^1,u) + 1$
        \State $\textsc{ComputeDistanceRec}(v)$ 
        \EndFor
    \end{algorithmic}
\end{algorithm}

\section{Proofs of Lemmas \ref{lm:st-color}, \ref{lm:st-update} and \ref{lm:st-get-closest}}\label{apx:st-proofs}

{\bf Lemma \ref{lm:st-color}} {\em
The request color procedure $\textsc{ColorRequest}$ is correct and has $O(\log d)$ time complexity.}
\begin{proof}
If the segment tree stores correct colors for segments, then the correctness of the algorithm follows from the description. The algorithm returns a color only if $x$ belongs to a segment that has a single color.
On each step, we change a node to a node on the next level. The tree is a full binary tree. Therefore, it has $h$ levels. Hence, the time complexity is $O(h)=O(\log d)$ because $2^{h-1}\leq d\leq 2^{h}$.  
\Endproof\end{proof}

\noindent
{\bf Lemma \ref{lm:st-update}} {\em
The update procedure $\textsc{ColorUpdate}$ is correct and has $O(\log d)$ time complexity.}
\begin{proof}
If the segment tree stores correct colors for segments, then the correctness of the algorithm follows from the description. The algorithm colors a required segment and keeps the color of the rest part.

Procedures $\textsc{ColorUpdatePrefix}$ and $\textsc{ColorUpdateSuffix}$ on each step change a node to a node on the next level. The tree is a full binary tree. Therefore, the tree has $h$ levels. Hence, the time complexity of these two algorithms is $O(h)=O(\log d)$ because $2^{h-1}\leq d\leq 2^{h}$. The procedure $\textsc{ColorUpdate}$ on each step changes a node to a node on the next level, then, stops and invokes the procedure $\textsc{ColorUpdatePrefix}$ and the procedure $\textsc{ColorUpdateSuffix}$. Its time complexity is  $O(h)=O(\log d)$ also. We can say that procedures run consistently. Therefore, the total time complexity is $O(\log d)$. 
\Endproof\end{proof}

\noindent
{\bf Lemma \ref{lm:st-get-closest}} {\em
$\textsc{GetClosestColorLeft}$ and $\textsc{GetClosestColorRight}$ work correct with $O(\log d)$ time complexity. }
\begin{proof}
The proof is similar to the proof of Lemma \ref{lm:st-update}.
\Endproof\end{proof}

\section{Algorithms for Coloration Problem on a Segment Tree}\label{apx:st-algos}

\begin{algorithm}[ht]
    \caption{$\textsc{ColorUpdatePrefix}(l,r,c,c',root)$. An operation of update color of a prefix segment $[l,r]$ by a color $c$. The operation is for a segment tree with the $root$ node as a root. $c'$ is a color for rest part of the segment of $root$. If $c'$ is not assigned, then $c'=0$} \label{alg:st-update-prefix}
    \begin{algorithmic}
        \State $v\gets root$
        \While{$v$ is not a leaf}
        \If{$c'=0$ and $C(v)\geq 1$}
        \State $c'\gets C(v)$
        \EndIf
        \State $Max(v)\gets\max(Max(v),r)$, $Min(v)\gets \textsc{Left}(v)$
        \State $u\gets \textsc{LeftChild}(v)$
        \State $w\gets \textsc{RightChild}(v)$
        \If{$r\leq \textsc{Right}(u)$}
        \If{$c'\geq 1$}
         \State $C(w)\gets c'$
        \EndIf
        \State $v\gets u$
        \EndIf
        \If{$r> \textsc{Right}(u)$}
        \State $C(u)\gets c$, $Min(u)\gets \textsc{Left}(u)$, $Max(u)\gets \textsc{Right}(u)$
        \State $v\gets w$
        \EndIf
        \EndWhile
        \State $C(v)\gets c$, $Min(v)\gets \textsc{Left}(v)$, $Max(v)\gets \textsc{Right}(v)$
    \end{algorithmic}
\end{algorithm}

\begin{algorithm}[ht]
    \caption{$\textsc{ColorUpdateSuffix}(l,r,c,c',root)$. An operation of update color of a suffix segment $[l,r]$ by a color $c$. The operation is for a segment tree with $root$ node as a root. $c'$ is a color for rest part of the segment of $root$. If $c'$ is not assigned, then $c'=0$} \label{alg:st-update-suffix}
    \begin{algorithmic}
        \State $v\gets root$
        \While{$v$ is not a leaf}
        \If{$c'=0$ and $C(v)\geq 1$}
        \State $c'\gets C(v)$
        \EndIf
         \State $Min(v)\gets\min(Min(v),l)$, $Max(v)\gets \textsc{Right}(v)$
        \State $u\gets \textsc{LeftChild}(v)$
        \State $w\gets \textsc{RightChild}(v)$
        \If{$l\geq \textsc{Left}(w)$}
        \If{$c'\geq 1$}
         \State $C(u)\gets c'$
        \EndIf
        \State $v\gets w$
        \EndIf
        \If{$l< \textsc{Left}(w)$}
        \State $C(w)\gets c$, $Min(w)\gets \textsc{Left}(w)$, $Max(w)\gets \textsc{Right}(w)$
        \State $v\gets u$
        \EndIf
        \EndWhile
        \State $C(v)\gets c$, $Min(v)\gets \textsc{Left}(v)$, $Max(v)\gets \textsc{Right}(v)$
    \end{algorithmic}
\end{algorithm}

\begin{algorithm}[ht]
    \caption{$\textsc{ColorUpdate}(l,r,c,root)$. An operation of update color of a segment $[l,r]$ by a color $c$. The operation is for a segment tree with $root$ node as a root.} \label{alg:st-update}
    \begin{algorithmic}
        \State $v\gets root$
        \State $c'\gets 0$
        \State $Split \gets False$
        \While{$v$ is not a leaf and $Split=False$}
        \If{$c'=0$ and $C(v)\geq 1$}
        \State $c'\gets C(v)$
        \EndIf
        \State $Min(v)\gets\min(Min(v),l)$, $Max(v)\gets\max(Max(v),r)$
        \State $u\gets \textsc{LeftChild}(v)$
        \State $w\gets \textsc{RightChild}(v)$
        \If{$l\geq \textsc{Left}(w)$}
        \If{$c'\geq 1$}
         \State $C(u)\gets c'$
        \EndIf
        \State $v\gets w$
        \EndIf
        \If{$r\leq \textsc{Right}(u)$}
        \If{$c'\geq 1$}
         \State $C(w)\gets c'$
        \EndIf
        \State $v\gets u$
        \EndIf
        \If{$l\leq \textsc{Right}(u)$ and $r\geq \textsc{Left}(w)$}
        \State $Split \gets True$
        \State $\textsc{ColorUpdateSuffix}(l,\textsc{Right}(u),c, c',u)$
        \State $\textsc{ColorUpdatePrfix}(\textsc{Left}(w),r,c, c',w)$
        \EndIf
        \EndWhile
        \If{$v$ is a leaf}
        \State $C(v)\gets c$
        \EndIf
    \end{algorithmic}
\end{algorithm}

\begin{algorithm}[ht]
    \caption{$\textsc{GetClosestColorRightPrefix}(l,r,root)$. A request for the minimal index of a colored element of a prefix segment $[l,r]$. It returns $NULL$ if there is no such elements} \label{alg:st-min-color-prefix}
    \begin{algorithmic}
        \State $v\gets root$
        \State $Result\gets NULL$
        \If{$Min(v)\neq 2^h+1$}
        \State $Found\gets False$
        \While{$v$ is not a leaf and $Found=False$}
            \State $u\gets \textsc{LeftChild}(v)$
            \State $w\gets \textsc{RightChild}(v)$
            \If{$r\leq \textsc{Right}(u)$}
            \State $v\gets u$
            \EndIf
            \If{$r\geq \textsc{Left}(w)$ and $Min(u)\neq 2^h+1$}
             \State $Result\gets Min(u)$
             \State $Found\gets True$
            \EndIf
            \If{$r\geq \textsc{Left}(w)$ and $Min(u)=2^h+1$}
             \State $v\gets w$
            \EndIf
        \EndWhile
        \If{$Found=False$ and $Min(v)\neq 2^h+1$}
        \State $Result = Min(v)$
        \EndIf
        \EndIf
        \State \Return{$Result$}
    \end{algorithmic}
\end{algorithm}

\begin{algorithm}[ht]
    \caption{$\textsc{GetClosestColorRightSuffix}(l,r,root)$. A request for the minimal index of a colored element of a suffix segment $[l,r]$. It returns $NULL$ if there is no such elements} \label{alg:st-min-color-suffix}
    \begin{algorithmic}
        \State $v\gets root$
        \State $Result\gets NULL$
        \If{$Min(v)\neq 2^h+1$}
        \State $Found\gets False$
        \While{$v$ is not a leaf and $Found=False$}
            \State $u\gets \textsc{LeftChild}(v)$
            \State $w\gets \textsc{RightChild}(v)$
            \If{$l\geq \textsc{Left}(w)$}
              \State $v\gets w$
            \EndIf
            \If{$l\leq \textsc{Right}(u)$ and $Min(u)\neq 2^h+1$}
              \State $v\gets u$
            \EndIf
            \If{$l\leq \textsc{Right}(u)$ and $Min(u)=2^h+1$}
            \State $Result\gets Min(w)$
             \State $Found\gets True$
            \EndIf
        \EndWhile
        \If{$Found=False$ and $Min(v)\neq 2^h+1$}
        \State $Result = Min(v)$
        \EndIf
        \EndIf
        \State \Return{$Result$}
    \end{algorithmic}
\end{algorithm}
\section{Proof of Theorem \ref{th:preproc}}\label{apx:th1}
{\bf Theorem \ref{th:preproc}} {\em
Algorithm \ref{alg:preproc} for preprocessing has time complexity $O(n)$.}
\begin{proof}
As it was mentioned in Section \ref{sec:graph} the time complexity of  Heavy-light decomposition ${\cal P}$ construction is $O(n)$. 

Due to Lemma \ref{lm:st-construct}, time complexity of $\textsc{ConstructST}(P)$ is $O(|P|)$. The total time complexity of constructing all segment trees is $O\left(\sum_{P\in{\cal P}}|P|\right)=O(n)$ because of property of the decomposition.

Time complexity of $\textsc{ComputeDistance}$ is $O(n)$. Therefore, the total time complexity is $O(n)$.
\Endproof\end{proof}

\section{Proof of Lemma \ref{lm:coloring}}\label{apx:coloring}
{\bf Lemma \ref{lm:coloring}} {\em
Time complexity of Algorithm \ref{alg:coloring} is $O\left((\log n)^2\right)$.}
\begin{proof}
Due to properties of Heavy-light decomposition from Section \ref{sec:graph}, $t, t'= O(\log n)$. Due to Lemma \ref{lm:st-update}, each invocation of  $\textsc{ColorUpdate}$ for $P$ has time complexity $O(\log |P|)=O(\log n)$. So, the total time complexity is $O\left((\log n)^2 \right)$. 
\Endproof\end{proof}

\end{document}